\documentclass[letterpaper, 10 pt, conference]{ieeeconf} 

\IEEEoverridecommandlockouts  
\usepackage{amsmath}
\usepackage{amssymb}
\usepackage{amsthm}
\usepackage{mathtools}
\usepackage{derivative}
\usepackage{xcolor}
\usepackage{bm}

\usepackage[noadjust]{cite}
\usepackage{url}

\usepackage{graphicx}
\usepackage[export]{adjustbox} 
\graphicspath{{figures/}}

\newtheorem{theorem}{Theorem}
\newtheorem{lemma}{Lemma}
\newtheorem{proposition}{Proposition}

\theoremstyle{definition}
\newtheorem{definition}{Definition}

\newtheorem{remark}{Remark}

\newtheorem{example}{Example}

\newcommand{\argmin}{\operatornamewithlimits{arg\,min}}

\newcommand{\R}{\mathbb{R}}

\newcommand{\C}{\mathcal{C}}

\newcommand{\T}{^\top}

\newcommand{\K}{\mathcal{K}}
\newcommand{\Kinf}{\mathcal{K}_{\infty}}

\newcommand{\relu}{\mathrm{ReLU}}
\newcommand{\QP}{\mathrm{QP}}
\newcommand{\map}[3]{#1:#2 \rightarrow #3}

\renewcommand{\bf}{\mathbf{f}}
\newcommand{\bff}{\mathbf{f}}
\newcommand{\bg}{\mathbf{g}}
\newcommand{\bk}{\mathbf{k}}

\newcommand{\bu}{\mathbf{u}}

\newcommand{\bx}{\mathbf{x}}

\newcommand{\Lf}{L_{\bff}}
\newcommand{\Lg}{L_{\bg}}
\newcommand{\LfV}{\Lf V}
\newcommand{\LgV}{\Lg V}
\newcommand{\Lfh}{\Lf h}
\newcommand{\Lgh}{\Lg h}

\title{\textbf{Characterizing Smooth Safety Filters via the Implicit Function Theorem}}

\author{Max H. Cohen, Pio Ong, Gilbert Bahati, and Aaron D. Ames %
\thanks{This research was supported by NSF Award \# 1932091.}
\thanks{The authors are with the Department of Mechanical and Civil Engineering, California Institute of Technology, Pasadena, CA 91125; \texttt{\{maxcohen,pioong,gbahati,ames\}@caltech.edu}}
}

\begin{document}
    \maketitle

    \begin{abstract}
        Optimization-based safety filters, such as control barrier function (CBF) based quadratic programs (QPs), have demonstrated success in controlling autonomous systems to achieve complex goals. These CBF-QPs can be shown to be continuous, but are generally not smooth, let alone continuously differentiable. In this paper, we present a general characterization of smooth safety filters -- smooth controllers that guarantee safety in a minimally invasive fashion -- based on the Implicit Function Theorem. This characterization leads to families of smooth universal formulas for safety-critical controllers that quantify the conservatism of the resulting safety filter, the utility of which is demonstrated through illustrative examples.
    \end{abstract}

    \section{Introduction}   
    Over the past decade, control barrier functions (CBFs) \cite{AmesTAC17} have proven to be a powerful tool for designing controllers enforcing safety on nonlinear systems. The properties of CBFs naturally lead to their use as \emph{safety filters} for nominal controllers that may not have been designed a priori to ensure safety. Most often, such safety filters are instantiated via optimization problems -- typically a quadratic program (QP) -- to minimize the deviation from a nominal controller while satisfying Lyapunov-like conditions that ensure forward invariance of a designated safe set \cite{AmesECC19,WeiBook,PanagouAutomatica23}. Under certain regularity conditions, these optimization-based safety filters are locally Lipschitz functions of the system state \cite{JankovicAutomatica18,CortesArXiV23}, allowing one to leverage set-theoretic tools, such as Nagumo's Theorem \cite{Blanchini}, to conclude forward invariance of safe sets. Although such controllers are pointwise optimal, they are typically not smooth even if the problem data is.

    \begin{figure}
        \centering
        \includegraphics[width=0.48\textwidth]{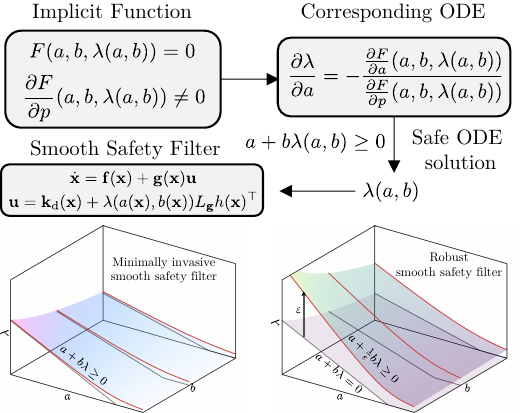}
        \caption{Illustration of methodology for generating smooth safety filters.}
        \label{fig:main_figure}\vspace{-0.5cm}
    \end{figure}

    The lack of smoothness exhibited by optimization-based safety filters has not been overly detrimental to the development of safety-critical controllers to date; however, more recent developments in the literature motivate the consideration of smooth (or, at least, sufficiently differentiable) safety filters. For example, \cite{TaylorCDC22} recently proposed a barrier-backstepping methodology that enables the systematic construction of CBFs for systems in strict-feedback form. As in Lyapunov backstepping \cite{Krstic}, such an approach requires differentiating through virtual CBF controllers at intermediate layers to construct a composite CBF for the overall system. Similar ideas are leveraged in \cite{AmesRAL22,AmesACC23} to construct CBFs for robotic systems based on reduced-order models -- an approach successfully used to safely control complex robotic systems, e.g., walking robots and drones. 

    Given the similarities between control Lyapunov functions (CLFs) and CBFs, one may wonder if it is possible to adapt smooth universal formulas for CLFs \cite{SontagSCL89} to CBFs. The answer is affirmative -- with some slight modifications.  Sontag's Universal Formula for stabilization \cite{SontagSCL89} can be applied to safety as noted in \cite{AmesLCSS19,KrsticTAC23,SunACC23,PolArXiV23}. Despite this, Sontag's formula is scarcely used as a safety filter since, in its most common form, such a controller tends to be overly invasive, overriding inputs from the nominal controller even when not necessary to ensure safety. Alternative smooth universal formulas have been proposed in \cite{OngCDC19} based on computing weighted centroids of the set of all control values satisfying CBF and/or CLF conditions using the probability density function of a Gaussian distribution. In a different approach, the authors of  \cite{JayawardhanaAutomatica16} leverage Sontag's formula to combine stabilization and safety objectives.  Yet questions remain around the connections between smoothness and safety filters.

    The main objective of this paper is to provide a general characterization of smooth safety filters -- smooth controllers that guarantee safety in a minimally invasive fashion. Our characterization is motivated by the original development of Sontag's formula in \cite{SontagSCL89}. Sontag's formula is derived by computing the roots of an algebraic equation parameterized by the Lie derivatives of a CLF/CBF. When certain regularity conditions are met, the smoothness of such roots as a function of the Lie derivatives follows directly from the Implicit Function Theorem. In this paper, we seek a deeper understanding of the properties of this equation.  
    In particular: What properties should this equation satisfy so that one of its solutions produces a smooth minimally invasive controller that guarantees safety? We answer this question by constructing an ordinary differential equation (ODE) from a given algebraic one such that the trajectories of this ODE coincide with the solutions of the algebraic equation. Leveraging invariance-like tools, we introduce sufficient conditions for this ODE so that its trajectories produce a smooth safety filter. This characterization leads to various smooth universal formulas for safety-critical control that allow one to assess the conservatism of the resulting safety filter, the utilities of which we illustrate through their application to safety-critical control based on reduced order models \cite{AmesRAL22,AmesACC23}.

    \section{Preliminaries and Problem Formulation}
    Consider the nonlinear control affine system:
    \begin{equation}\label{eq:dyn}
        \dot{\bx} = \bff(\bx) + \bg(\bx)\bu,
    \end{equation}
    where $\bx\in\R^n$ is the system state, $\bu\in\R^m$ is the control input, $\bff\,:\,\R^n\rightarrow\R^n$ is the drift vector field, and $\bg\,:\,\R^n\rightarrow\R^{n\times m}$ captures the control directions. Throughout this paper, we assume that $\bff$ and $\bg$ are smooth functions of the state. Applying a smooth feedback controller $\bk\,:\,\R^n\rightarrow\R^m$ to \eqref{eq:dyn} produces the smooth closed-loop system:
    \begin{equation}\label{eq:dyn-cl}
        \dot{\bx} = \bff(\bx) + \bg(\bx)\bk(\bx),
    \end{equation}
    which, for each initial condition $\bx_0\in\R^n$, generates a unique smooth solution $\bx\,:\,I(\bx_0)\rightarrow\R^n$ satisfying \eqref{eq:dyn-cl} on some maximal interval of existence $I(\bx_0)\subseteq\R_{\geq0}$. 

    \subsection{CLFs and Sontag's Universal Formula}
    Before discussing smooth safety filters, we recount the main ideas behind CLFs as presented in \cite{SontagSCL89}. Recall that a smooth, proper, positive definite function $V\,:\,\R^n\rightarrow\R_{\geq0}$ is a CLF for \eqref{eq:dyn} if for all $\bx\in\R^n\setminus\{0\}$
    \begin{equation*}
        \inf_{\bu\in\R^m}\{\underbrace{\nabla V(\bx)\cdot\bff(\bx)}_{L_{\bff}V(\bx)} + \underbrace{\nabla V(\bx)\cdot \bg(\bx)}_{L_{\bg}V(\bx)}\bu \} < 0.
    \end{equation*}
    The existence of a CLF implies that for each $\bx\in\R^n$ there exists a $\bu\in\R^m$ that enforces $V$ to decrease, and allows for constructing a feedback controller $\bk\,:\,\R^n\rightarrow\R^m$ that renders the origin asymptotically stable by ensuring that:
    \begin{equation}\label{eq:V-dot}
        \forall \bx\in\R^n\setminus\{0\}\,:\,\LfV(\bx) + \LgV(\bx)\bk(\bx)<0.
    \end{equation}
    In this paper, we are concerned with designing \emph{smooth} feedback controllers satisfying a general class of affine inequalities, such as the one in \eqref{eq:V-dot}. In \cite{SontagSCL89} Sontag provides one example of such a controller, now known as Sontag's Universal Formula for stabilization, which is given by:
    \begin{equation}\label{eq:k-CLF}
        \bk(\bx) = \lambda_{\mathrm{CLF}}(\LfV(\bx), \|\LgV(\bx)\|^2)\LgV(\bx)\T,
    \end{equation}
    \begin{equation}\label{eq:lambda-Sontag-CLF}
        \lambda_{\mathrm{CLF}}(a,b) \coloneqq
        \begin{cases}
            0 & b = 0 \\
            \frac{-a - \sqrt{a^2 + q(b)b}}{b} & b\neq0,
        \end{cases}
    \end{equation}
    where $q\,:\,\R\rightarrow\R$ is smooth and satisfies $q(0)=0$ and $q(b)>0$ for all $b\neq0$. In \cite{SontagSCL89}, the smoothness of Sontag's formula \eqref{eq:lambda-Sontag-CLF} is proven using an argument based on the Implicit Function Theorem \cite[Thm. 11.2]{Loomis}. Specifically, the following result \cite[Thm. 11.1]{Loomis}, which is related to the Implicit Function Theorem, is useful for establishing smoothness of \eqref{eq:lambda-Sontag-CLF}.
    \begin{theorem}[\cite{Loomis}]\label{theorem:implicit-function-theorem}
        Let $(a,b,p)\mapsto F(a,b,p)$ be a smooth function. If a continuous function $\lambda\,:\,\mathcal{S}\rightarrow\R$ satisfies:
        \begin{equation}\label{eq:F=0}
            F(a,b,\lambda(a,b))=0,
        \end{equation}
        \begin{equation}\label{eq:dFdp}
            \pdv{F}{p}(a,b,\lambda(a,b))\neq0,
        \end{equation}
        for all $(a,b)\in\mathcal{S}$, then $\lambda$ is smooth for all $(a,b)\in\mathcal{S}$ and its derivative is given by:
        \begin{equation}\label{eq:dF}
            \begin{bmatrix}
                \pdv{\lambda}{a}(a,b) \\ \pdv{\lambda}{b}(a,b)
            \end{bmatrix}
            =
            -\frac{1}{\pdv{F}{p}(a,b,\lambda(a,b))}
            \begin{bmatrix}
                \pdv{F}{a}(a,b,\lambda(a,b)) \\ \pdv{F}{b}(a,b,\lambda(a,b))
            \end{bmatrix}. 
        \end{equation}
    \end{theorem}
    To apply this theorem and show smoothness of the function $\lambda_{\mathrm{CLF}}$ in \eqref{eq:lambda-Sontag-CLF}, one considers the smooth function:
    \begin{equation}\label{eq:F-Sontag}
        F(a,b,p) = bp^2 + 2ap - q(b),
    \end{equation}
    noting that $\lambda_{\mathrm{CLF}}$ is continuous and satisfies \eqref{eq:F=0} and \eqref{eq:dFdp} for each $(a,b)$ in $\mathcal{S}_{\mathrm{CLF}} \coloneqq \{(a,b)\in\R\times\R_{\geq0}\,:\,a<0\,\vee\,b>0\}$. 
    In addition to being smooth, one can also verify that the controller \eqref{eq:k-CLF} constructed from $\lambda_\mathrm{CLF}$ satisfies inequality~\eqref{eq:V-dot}. Indeed, this is a consequence of picking an appropriate function for $F$, a point which we will expand on later.
    
    \subsection{CBFs and Safety Filters}
    The main objective of this paper is to leverage Theorem~\ref{theorem:implicit-function-theorem} for constructing smooth controllers that will render the resulting closed-loop system \emph{safe}, a property often formalized using the framework of set invariance \cite{Blanchini}. Formally, we say that \eqref{eq:dyn-cl} is safe on a set $\C\subset\R^n$ if $\C$ is forward invariant. By considering sets of the form:
    \begin{equation}\label{eq:C}
        \C = \{\bx\in\R^n\,:\,h(\bx) \geq 0\},
    \end{equation}
    where $h\,:\,\R^n\rightarrow\R$ is a smooth function, the existence of a safe feedback controller can be characterized using the concept of a control barrier function (CBF) \cite{AmesTAC17}.

    \begin{figure*}
        \centering
        \includegraphics{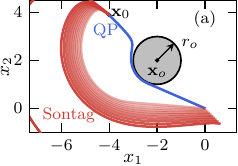}
        \hfill
        \includegraphics{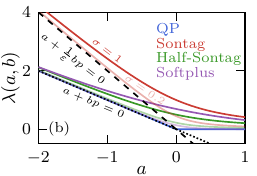}
        \hfill
        \includegraphics{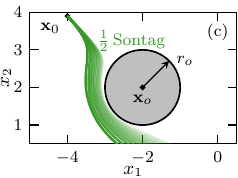}
        \hfill
        \includegraphics{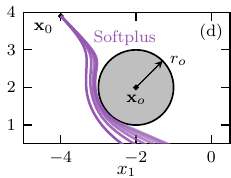}
        \vspace{-0.3cm}
        \caption{(a) Comparison between the trajectory of a single integrator generated by the QP controller (blue) in \eqref{eq:cbf-QP} and Sontag safety filter (red) from \eqref{eq:sontag-cbf} for Example \ref{ex:single-int}. Each safety filter is implemented from the initial condition $\bx_0=(-4,3.9)$ with $\alpha=2$. The values of $\sigma$ are varied from 0.2 to 0.001. (b) Illustration of the sets $\mathcal{H}_b$ (dotted black line) and $\mathcal{H}_b^{\varepsilon}$ (dashed black line) from \eqref{eq:Hb}  and \eqref{eq:H-epsilon}, respectively, for a fixed $b$. (c-d) Trajectories of the single integrator under the influence of the smooth safety filters from \eqref{eq:lambda-half-sontag} (c) and \eqref{eq:lambda-softplus} (d). The trajectories are generated with $\sigma$ varying from 1 to 0.01 for (c) and from 0.7 to 0.01 for (d). In each plot, more transparent curves correspond to smaller values of $\sigma$.}
        \label{fig:simple-example}\vspace{-0.5cm}
    \end{figure*}
    
    \begin{definition}
        A smooth function $h\,:\,\R^n\rightarrow\R$ defining a set $\C\subset\R^n$ as in \eqref{eq:C} is said to be a \emph{control barrier function} (CBF) for \eqref{eq:dyn} on $\C$ if zero is a regular value of $h$ and there exists a smooth\footnote{A continuous function $\alpha\,:\,\R\rightarrow\R$ is an extended class $\Kinf$ function ($\alpha\in\Kinf^e$) if $\alpha(0)=0$, $\alpha$ is increasing, and $\lim_{r\rightarrow\pm\infty}\alpha(r)=\pm\infty$.} $\alpha\in\K_{\infty}^e$ such that for all $\bx\in\R^n$
        \begin{equation*}\label{eq:cbf}
            \sup_{\bu\in\R^m}\big\{\underbrace{\nabla h(\bx)\cdot\bff(\bx)}_{L_{\bff}h(\bx)} + \underbrace{\nabla h(\bx)\cdot\bg(\bx)}_{L_{\bg}h(\bx)}\bu \big\} > -\alpha(h(\bx)).
        \end{equation*}
    \end{definition}
    Similar to CLFs, any feedback controller $\bk\,:\,\R^n\rightarrow\R^m$ satisfying the inequality $\Lfh(\bx) + \Lgh(\bx)\bk(\bx)\geq-\alpha(h(\bx))$ ensures that $h$ remains positive along closed-loop trajectories, and therefore renders $\C$ forward invariant \cite{AmesTAC17}.
    Perhaps the greatest utility of CBFs is their ability to act as a \emph{safety filter} for a nominal feedback controller $\bk_{\mathrm{d}}\,:\,\R^n\rightarrow\R^m$.
    A safety filter is a controller that modifies $\bk_{\mathrm{d}}$ -- preferably, in a minimally invasive fashion -- so that the resulting closed-loop system is safe. The most common examples of such safety filters are instantiated via the QP:
    \begin{equation}\label{eq:cbf-QP}
        \begin{aligned}
            \bk(\bx) =\argmin_{\bu\in\R^m} \quad & \tfrac{1}{2}\|\bu - \bk_{\mathrm{d}}(\bx)\|^2 \\ 
            \text{subject to} \quad & L_{\bff}h(\bx) + L_{\bg}h(\bx)\bu \geq - \alpha(h(\bx)),
        \end{aligned}
    \end{equation} 
    which modifies $\bk_{\mathrm{d}}$ in the $\Lgh^\top$ direction:  
    \begin{equation}\label{eq:k-CBF}
        \bk(\bx) = \bk_{\mathrm{d}}(\bx) + \lambda(a(\bx),b(\bx))\Lgh(\bx)\T,
    \end{equation}
    with a scalar function $\map{\lambda}{\R\times\R}{\R}$, where:
    \begin{equation}\label{eq:ab}
        \begin{aligned}
            a(\bx) \coloneqq & L_{\bff}h(\bx) + L_{\bg}h(\bx)\bk_{\mathrm{d}}(\bx) + \alpha(h(\bx)) \\ 
            b(\bx) \coloneqq & \|L_{\bg}h(\bx)\|^2. \\
        \end{aligned}
    \end{equation}
    For the QP in \eqref{eq:cbf-QP}, $\lambda$ is the Lagrange multiplier associated with the constraint and is given by:
    \begin{equation}\label{eq:lambda-QP}
        \lambda(a,b)=\lambda_{\QP}(a,b) \coloneqq 
        \begin{cases}
            0 & b = 0 \\ 
            \relu(-a/b) & b > 0,
        \end{cases}
    \end{equation}
    where $\relu(y)\coloneqq \max\{0,y\}$. This QP-based controller has the advantage of being pointwise optimal, but is not smooth even if the problem data itself is smooth. Before proceeding, we make precise the notion of a smooth safety filter.
    \begin{definition}\label{def:smooth-safety}
        Given a CBF $h\,:\,\R^n\rightarrow\R$ and nominal controller $\bk_{\mathrm{d}}\,:\,\R^n\rightarrow\R^m$ for \eqref{eq:dyn}, a controller $\bk\,:\,\R^n\rightarrow\R^m$ of the form \eqref{eq:k-CBF} is said to be a \emph{smooth safety filter} for \eqref{eq:dyn} with respect to $\C$ if $\bk$ is smooth and for all $\bx\in\R^n$
        \begin{equation}\label{eq:ab-constraint}
            \underbrace{a(\bx) + b(\bx)\lambda(a(\bx), b(\bx))}_{\Lfh(\bx) + \Lgh(\bx)\bk(\bx) + \alpha(h(\bx))} \geq 0.
        \end{equation}
    \end{definition}
    In the following section, we present our characterization of smooth safety filters.

    \section{Towards Smooth Safety Filters}\label{sec:smooth-safety}
    A Sontag-like safety filter can also be derived using Theorem \ref{theorem:implicit-function-theorem}. Here, we use the same function as in \eqref{eq:F-Sontag}, taking the other root to the quadratic function:
    \begin{equation}\label{eq:sontag-cbf}
        \lambda_{\mathrm{S}}(a,b) \coloneqq \begin{cases}
            0 & b = 0 \\ 
            \frac{-a + \sqrt{a^2 + q(b)b}}{b} & b > 0.
        \end{cases}
    \end{equation}
    One can verify $\lambda_{\mathrm{S}}$ is continuous and satisfies \eqref{eq:F=0} and \eqref{eq:dFdp} on
    \begin{equation}\label{eq:S}
        \mathcal{S} = \{(a,b)\in\R\times\R_{\geq0}\,:\,a>0\,\vee\,b>0\},
    \end{equation}
    implying $\lambda_{\mathrm{S}}$ is smooth on \eqref{eq:S} by Theorem~\ref{theorem:implicit-function-theorem} and therefore produces a smooth safety filter via \eqref{eq:k-CBF} as one can verify $a + b\lambda_{\mathrm{S}}(a,b)\geq0$. Nevertheless, the resulting safety filter is overly conservative as illustrated by the following example.
    \begin{example}\label{ex:single-int}
        Consider a single integrator $\dot{\bx}=\bu$ tasked with reaching the origin while avoiding a circular obstacle of radius $r_o\in\R_{>0}$ located at $\bx_o\in\R^2$. This task can be accomplished by designing a safety filter for the nominal controller $\bk_{\mathrm{d}}(\bx)=-\bx$ using $h(\bx)=\|\bx-\bx_o\|^2 - r_o^2$ as a CBF, which we construct using the QP in \eqref{eq:cbf-QP} and the Sontag safety filter defined by \eqref{eq:k-CBF} and \eqref{eq:sontag-cbf} with $q(b)=\sigma b$, $\sigma\in\R_{>0}$. The trajectory of the single integrator under each safety filter is provided in Fig. \ref{fig:simple-example}(a) for different $\sigma$. Decreasing $\sigma$ decreases the conservatism of the controller, but it is still overly invasive even for arbitrarily small $\sigma$.
    \end{example}
    Motivated by the previous example, we set out to develop a general class of smooth safety filters that are less invasive by analyzing how the choice of $F$ impacts the behavior of the resulting safety filter. Towards this development, we aim to answer the following: What properties should a function $(a,b,p)\mapsto F(a,b,p)$ satisfy so that one of its roots produces a smooth safety filter, and how may this function be selected such that the resulting safety filter is minimally invasive?

    \subsection{Minimally Invasive Smooth Safety Filters}
    To answer the question posed in the previous subsection, let $(a,b,p)\mapsto F(a,b,p)$ be a smooth function and suppose there exists a continuous function $\lambda\,:\,\mathcal{S}\rightarrow\R$ satisfying the conditions of Theorem \ref{theorem:implicit-function-theorem}, which implies that $\lambda$ is smooth on $\mathcal{S}$. We are also interested in ensuring that:
    \begin{equation}\label{eq:lambda-ab}
        a + b\lambda(a,b) \geq 0,
    \end{equation}
    for all $(a,b)\in\mathcal{S}$ so that $\lambda$ may be used to construct a smooth safety filter satisfying \eqref{eq:ab-constraint}. Note that $F$ in~\eqref{eq:F-Sontag} is constructed such that $\pdv{F}{p}=2(a + b\lambda_{\mathrm{S}}(a,b)) >0$ directly implies the satisfaction of~\eqref{eq:lambda-ab}. However, we demonstrate that this is not the only path towards certifying that one of $F$'s roots produces a smooth safety filter. In what follows, we introduce more general sufficient conditions on $F$ so that one of its roots produces a smooth safety filter.
    
    Our starting point is one of the direct consequences of Theorem \ref{theorem:implicit-function-theorem} --  the function $\lambda$ must satisfy:
    \begin{equation}\label{eq:lambda-dyn}
        \pdv{\lambda}{a}(a,b) = -\frac{\pdv{F}{a}(a,b,\lambda(a,b))}{\pdv{F}{p}(a,b,\lambda(a,b))},
    \end{equation}
    for all $(a,b)\in\mathcal{S}$, implying it is a solution to the ODE:
    \begin{equation}\label{eq:p-dyn}
        \odv{p}{a} = -\frac{\pdv{F}{a}(a,b,p)}{\pdv{F}{p}(a,b,p)},
    \end{equation}
    from an appropriate initial condition\footnote{In \eqref{eq:p-dyn}, one may think of $p$ as the state, $a$ as ``time", and $b$ as a parameter.}. For a fixed $b = b^*\geq0$, the trajectory $a\mapsto\lambda(a,b^*)$ defines a curve $\{(a,b,p)\in\mathcal{S}\times\R\,:\,p=\lambda(a,b),~b=b^*\}$. As $b$ is varied, this curve defines an entire surface $\{(a,b,p)\in\mathcal{S}\times\R\,:\,p=\lambda(a,b)\}$, thereby recovering $(a,b)\mapsto\lambda(a,b)$ as a function of both $a$ and $b$ (cf. Fig. \ref{fig:main_figure}).  
    Although working with an ODE is generally more challenging than working with an algebraic equation, this new perspective allows us to reformulate the goal of satisfying inequality~\eqref{eq:lambda-ab} as a set invariance problem. To do this, we define, for each $b\geq0$, the set-valued map\footnote{One may think of \eqref{eq:Hb} as a collection of ``time-varying" safe sets (recall that $a$ is our ``time" variable) for the dynamics in \eqref{eq:lambda-dyn}.}
    \begin{equation}\label{eq:Hb}
        \mathcal{H}_b(a)\coloneqq \{p\in\R\,:\,a + bp\geq0\}.
    \end{equation}
    Ultimately, we would like to ensure that for each fixed $b=b^*\geq 0$, the solution to~\eqref{eq:lambda-dyn} satisfies $\lambda(a,b^*)\in\mathcal{H}_{b^*}(a)$ for all $(a,b^*)\in \mathcal{S}$ so that we can conclude that $\lambda$ satisfies \eqref{eq:lambda-ab} for all $(a,b)\in\mathcal{S}$. The following proposition constitutes the first result of this paper, establishing conditions on $F$ such that the flow of \eqref{eq:p-dyn} satisfies inequality \eqref{eq:lambda-ab}.
    \begin{proposition}\label{proposition:smooth-safety-1}
        Let $F\,:\,\R^3\rightarrow\R$ be a smooth function 
        that defines the dynamics in \eqref{eq:p-dyn}. Suppose that for all $b>0$:
        \begin{equation}\label{eq:invariance-condition-1}
            \pdv{F}{p}\bigg\vert_{p=-\frac{a}{b}} = b\pdv{F}{a}\bigg\vert_{p=-\frac{a}{b}},
        \end{equation}
        Then, for all $b>0$, the set $\mathcal{H}_b$ is invariant for \eqref{eq:p-dyn}.
    \end{proposition}
    \begin{proof}
        We begin by noting that the dynamics~\eqref{eq:p-dyn} may not be well-defined for a given initial condition $p_0(a_0)\in\mathcal H_b(a_0)$ since we may have $\pdv{F}{p}(a_0,b,p_0)=0$. 
        Nevertheless, there exists no solution\footnote{We consider solutions in the classical sense.} from such initial conditions, and the empty solution is contained in $\mathcal H_b$.

        For other initial conditions $p_0(a_0)\in \mathcal H_b(a_0)$ where the dynamics are well-defined, there exists an interval of existence $I(p_0(a_0))\subset\R$ for the solution $\lambda(\cdot,b)\,:\,I(p_0(a_0))\rightarrow\R$ from each initial condition.
        We now show that when $b>0$, the condition $a+b\lambda(a,b)\geq 0$ holds for all $a\in I(p_0(a_0))$. To this end, we define:
        \begin{equation}\label{eq:hb}
            h_b(p,a) = a + bp.
        \end{equation}
        Because $\pdv{h_b}{p}(p,a)=b>0$, zero is a regular value of $h_b$; hence, the condition $h_b(\lambda(a,b),a) \geq 0$ holds along the trajectory provided that the total derivative $\odv{h_b}{a} = 0$ whenever $h_b(p,a)=0$, i.e., when $p=-a/b$. Note that we use equality in the above (rather than an inequality) since we must show invariance of $\mathcal{H}_b$, not just forward invariance, since $a$ may take any value in $\R$.
        Computing the derivative of $h_b$ along the trajectory of \eqref{eq:p-dyn} yields:
        \begin{equation*}
                \odv{h_b}{a} =  \pdv{h}{a} + \pdv{h}{p}\odv{p}{a} 
                =  1 - b\bigg[\pdv{F}{p}\bigg]^{-1} \pdv{F}{a}.
        \end{equation*}
        It follows from \eqref{eq:invariance-condition-1} that the above evaluates to zero when $h_b(p,a)=0$ so long as $\pdv{F}{p}\neq0$, which must be true for the solution to exist, and implies $\mathcal H_b$ is invariant, as desired.
    \end{proof}

    Proposition \ref{proposition:smooth-safety-1} provides a simple condition~\eqref{eq:invariance-condition-1}, for which one can use to help assess if an implicit function $\lambda$ to a given smooth function $F$ will satisfy inequality~\eqref{eq:lambda-ab}. The key idea is that with Proposition~\ref{proposition:smooth-safety-1}, if $\lambda$ satisfies \eqref{eq:lambda-ab} for some $(a^*,b^*)\in \mathcal{S}$, it will satisfy \eqref{eq:lambda-ab} for all $a$ in its interval of existence for the same $b=b^*$. The following lemma establishes this idea formally and additional conditions on $\lambda$ so that it satisfies \eqref{eq:lambda-ab} for all $(a,b)\in\mathcal{S}$. 
    \begin{lemma}\label{lemma:invariance}
        Let $F\,:\,\R^3\rightarrow\R$ be a smooth function satisfying \eqref{eq:invariance-condition-1} for all $b>0$, and suppose there exists a continuous function $\lambda\,:\,\mathcal{S}\rightarrow\R$ satisfying~\eqref{eq:F=0} and \eqref{eq:dFdp} for all $(a,b)\in \mathcal S$ as in \eqref{eq:S}. Then, if $\lambda(0,b)>0$ for all $b>0$, $\lambda$ is smooth on $\mathcal{S}$ and satisfies inequality \eqref{eq:lambda-ab} for all $(a,b)\in\mathcal S$.
    \end{lemma}
    \begin{proof}
        The proof is divided into two cases: i) $b>0$ and ii) $b=0$. For each $b>0$, consider an initial condition of \eqref{eq:p-dyn}, with $a_0=0$, satisfying $p_0^b=\lambda(0,b)>0$. Since $F$ and $\lambda$ satisfy \eqref{eq:F=0} and \eqref{eq:dFdp}, the conditions of Theorem \ref{theorem:implicit-function-theorem} hold, which implies $\lambda$ is smooth on $\mathcal{S}$. Moreover, $\lambda(a,b)$ is the unique solution to the ODE \eqref{eq:p-dyn} by definition since, by Theorem \ref{theorem:implicit-function-theorem}, it must satisfy \eqref{eq:dF} for all $(a,b)\in\mathcal{S}$. Thus, the maximal interval of existence of this solution must be the domain of $\lambda$, which is equal to $\R$ when $b>0$. Moreover, since $p_0^b \in \mathcal H_b(0)$ by definition and $F$ satisfies \eqref{eq:invariance-condition-1}, the conditions of Proposition \ref{proposition:smooth-safety-1} hold, which implies $\lambda(a,b)\in\mathcal{H}_b(a)$ for all $a\in\R$ for each $b>0$. When $b=0$ any value of $\lambda(a,b)$ satisfies $\lambda(a,b)\in\mathcal{H}_b(a)$ for all $a>0$. Thus, $\lambda(a,b)\in\mathcal{H}_{b}(a)$ for all $(a,b)\in\mathcal{S}$, which implies $\lambda$ satisfies \eqref{eq:lambda-ab}, as desired.
    \end{proof}
    Lemma~\ref{lemma:invariance} suggests that with condition~\eqref{eq:invariance-condition-1} from Proposition \ref{proposition:smooth-safety-1}, $F$ only needs to be constructed so that $\lambda$ is positive when $a=0$ for it to generate a smooth safety filter. Note that the above result does not guarantee the existence of a continuous $\lambda\,:\,\mathcal{S}\rightarrow\R$ satisfying \eqref{eq:F=0} and \eqref{eq:dFdp}, but states that, if such a function exists, then it is smooth on $\mathcal{S}$ and satisfies inequality \eqref{eq:lambda-dyn}. We will provide examples of $F$ satisfying these conditions shortly. First, we combine Proposition \ref{proposition:smooth-safety-1} and Lemma \ref{lemma:invariance} to establish the main result of this paper, which formalizes the construction of smooth safety filters.

    \begin{theorem}\label{theorem:smooth-safety}
        Consider system \eqref{eq:dyn} with a smooth nominal controller $\map{\bk_{\mathrm{d}}}{\R^n}{\R^m}$ and let $\map{h}{\R^n}{\R}$ be a CBF for \eqref{eq:dyn} on a set $\C\subset\R^n$ as in \eqref{eq:C}. Let $\lambda\,:\,\mathcal{S}\rightarrow\R$ satisfy the conditions of Lemma \ref{lemma:invariance} for some smooth function $\map{F}{\R^3}{\R}$. Then, the controller:
        \begin{equation}\label{eq:smooth-safety-filter}
            \bk_{\mathrm{s}}(\bx) = \bk_{\mathrm{d}}(\bx) + \lambda(a(\bx),b(\bx))\Lgh(\bx)\T,
        \end{equation}
        where $a\,:\,\R^n\rightarrow\R$ and $b\,:\,\R^n\rightarrow\R$ are as in \eqref{eq:ab}, is a smooth safety filter for \eqref{eq:dyn}. 
    \end{theorem}
    \begin{proof}
    The proof follows directly from Lemma \ref{lemma:invariance} since $\lambda$ is smooth and satisfies \eqref{eq:lambda-ab}, implying \eqref{eq:smooth-safety-filter} satisfies \eqref{eq:ab-constraint}.
    \end{proof}
    
    \begin{example}\label{ex:example-F}
        We use our results to construct the functions:
        \begin{subequations}\label{eq:F-examples}
            \begin{equation}\label{eq:F-half-sontag}
                F_1(a,b,p) = bp^2 + ap -\tfrac{1}{4}q(b),
            \end{equation}
            \begin{equation}\label{eq:F-softplus}
                F_2(a,b,p) = e^{{\frac{p}{\sigma}}} - e^{-\frac{a}{\sigma b}} - 1,
            \end{equation}
        \end{subequations}
        where $q$ is as in \eqref{eq:F-Sontag} and $\sigma\in\R_{>0}$. These functions satisfy \eqref{eq:invariance-condition-1} and the conditions of Lemma \ref{lemma:invariance} where: 
        \begin{subequations}
            \begin{equation*}
                F_1(0,b,p) = bp^2 - \tfrac{1}{4}q(b) = 0 \implies \lambda(0,b) = \pm \tfrac{1}{2}\sqrt{\tfrac{q(b)}{b}},
            \end{equation*}
            \begin{equation*}
                F_2(0,b,p) = e^{\frac{p}{\sigma}} -2 = 0 \implies \lambda(0,b) = \sigma\log(2).
            \end{equation*}
        \end{subequations}
        Hence, for each $F$ there exists a continuous $\lambda$ satisfying $\lambda(0,b)>0$ for all $b>0$ (recall that $q(b)>0$ for all $b>0$). After verifying these conditions, we proceed to compute the implicit functions satisfying \eqref{eq:F=0}:
        \begin{subequations}\label{eq:lambda-examples}
            \begin{equation}\label{eq:lambda-half-sontag}
                \lambda_1(a,b)=\tfrac{1}{2}\lambda_{\mathrm{S}}(a,b)=\begin{cases}
                    0 & b=0 \\ 
                    \frac{-a + \sqrt{a^2 + bq(b)}}{2b} & b>0
                \end{cases}
            \end{equation}
            \begin{equation}\label{eq:lambda-softplus}
                \lambda_2(a,b) = \begin{cases}
                    0 & b=0 \\ 
                    \sigma\log(1 + e^{-\frac{a}{b\sigma}}) & b>0,
                \end{cases}
            \end{equation}
        \end{subequations}
        which are continuous on $\mathcal{S}$, and therefore smooth on $\mathcal{S}$ as one can verify that $\pdv{F}{p}(a,b,\lambda(a,b))\neq0$ for all $(a,b)\in\mathcal{S}$. These functions are plotted in Fig. \ref{fig:simple-example}(b) for a fixed $b$. As guaranteed by Lemma~\ref{lemma:invariance}, each $\lambda$ satisfies inequality~\eqref{eq:lambda-ab}. Moreover, when $q(b)=\sigma b$ both \eqref{eq:lambda-half-sontag} and \eqref{eq:lambda-softplus} approach $\lambda_{\QP}$ in \eqref{eq:lambda-QP} in the limit as $\sigma\rightarrow0$. Indeed, when $b>0$ both \eqref{eq:lambda-half-sontag} and \eqref{eq:lambda-softplus} are smooth approximations of the $\relu$ function, corresponding to the Squareplus approximation (yielding a ``Half-Sontag" formula) and the Softplus approximation, respectively \cite{Squareplus}. The functions in \eqref{eq:lambda-examples} are used to construct smooth safety filters 
        via Theorem \ref{theorem:smooth-safety}, and are applied to the scenario from Example \ref{ex:single-int}, cf. Fig.~\ref{fig:simple-example}(c-d).
    \end{example}

    \subsection{Robust Smooth Safety Filters}
    The results in the previous subsection provide sufficient conditions under which the solution $(a,b)\mapsto\lambda(a,b)$ of \eqref{eq:p-dyn} satisfies inequality \eqref{eq:lambda-ab} for all $(a,b)\in\mathcal{S}$. Such conditions require $\mathcal{H}_b$ to be invariant for \eqref{eq:p-dyn}, which precludes the consideration of safety filters that remain in a strict subset of $\mathcal{H}_b$. For example, even though Sontag's formula in \eqref{eq:sontag-cbf} satisfies inequality \eqref{eq:lambda-ab} for all $(a,b)\in\mathcal{S}$, one can verify that the $F$ producing this formula \eqref{eq:F-Sontag} does not satisfy \eqref{eq:invariance-condition-1}.
    This is because the formula is contained within the set:
    \begin{equation}\label{eq:H-epsilon}
        \mathcal{H}^{\varepsilon}_b(a) \coloneqq \{p\in\R\,:\,a + \tfrac{1}{\varepsilon}bp \geq 0\},
    \end{equation}
    for $\varepsilon=2$, cf. Fig \ref{fig:simple-example}(b). The above set satisfies $\mathcal{H}^{\varepsilon}_b(a)\subseteq\mathcal{H}_b(a)$ for all $a\leq0$ and $\varepsilon\geq1$. When $a>0$, $\mathcal{H}^{\varepsilon}_b(a)\nsubseteq\mathcal{H}_b(a)$ for any $\varepsilon\geq1$. In this situation, the fact that Sontag's formula satisfies inequality \eqref{eq:lambda-ab} relies on the fact that it is always positive, $\lambda(a,b)\geq0$. Motivated by this observation, in this subsection we study the invariance of $\mathcal{H}_b^\varepsilon\cap\R_{\geq0}$, noting that:
    \begin{equation*}
        \varepsilon\geq 1 \implies \mathcal{H}_b^\varepsilon(a)\cap\R_{\geq0}\subset \mathcal{H}_b(a),
    \end{equation*}
    for all $(a,b)\in\mathcal{S}$. Taking this intersection imposes the additional requirement that $\lambda(a,b)\geq0$ for all $(a,b)\in\mathcal{S}$, which is not restrictive since $\lambda$ negative values of $\lambda$ imply the resulting controller is pushing in the wrong direction (i.e., toward the constraint boundary) and is attempting to violate \eqref{eq:lambda-ab}.
    As illustrated in Fig. \ref{fig:main_figure} and Fig. \ref{fig:simple-example}(b), increasing $\varepsilon$ lifts the boundary of $\mathcal{H}^{\varepsilon}_b$ off that of $\mathcal{H}_b$ leading to a more restricted set of values that $\lambda(a,b)$ may achieve. Although increasing $\varepsilon$ imposes a more conservative condition on $\lambda$, it adds an additional robustness margin to the resulting safety filter, which, as demonstrated in Sec. \ref{sec:examples}, may be useful in practice.
    The following proposition establishes conditions on $F$ such that the flow of \eqref{eq:lambda-dyn} satisfies the tightened condition in \eqref{eq:H-epsilon}.
    \begin{proposition}\label{proposition:smooth-safety-2}
       Let $F\,:\,\R^3\rightarrow\R$ be a smooth function that defines the dynamics in \eqref{eq:p-dyn}. Suppose that for all $b>0$:
        \begin{equation}\label{eq:invariance-condition-2}
           \varepsilon\pdv{F}{p}\Big\vert_{p=-\frac{\varepsilon a}{b}} = b\pdv{F}{a}\Big\vert_{p=-\frac{\varepsilon a}{b}},\quad \pdv{F}{a}\Big\vert_{p=0} = 0,
        \end{equation}
        Then, for each $b>0$, the set $\mathcal{H}_b^\varepsilon\cap\R_{\geq0}$ is invariant for \eqref{eq:p-dyn}.
    \end{proposition}
    \begin{proof}
        Showing the invariance of $\mathcal{H}_b^\varepsilon$ follows the same steps as that of Proposition \ref{proposition:smooth-safety-1} by replacing $h_b$ from \eqref{eq:hb} with $h_b^\varepsilon(p,a) \coloneqq a + \tfrac{1}{\varepsilon}bp$,
        which defines $\mathcal{H}_{b}^\varepsilon(a)$ as its zero superlevel set. To show that $\R_{\geq0}=\{p\in\R\,:\,p\geq0\}$ is invariant, we define $h_p(p)\coloneqq p$, which defines $\R_{\geq0}$ as its zero superlevel set and satisfies $\pdv{h_p}{p}\neq0$, implying zero is a regular value of $h_p$. Hence, $\R_{\geq0}$ is invariant provided that:
        \begin{equation*}
            \odv{h_p}{a} = \odv{p}{a} = -\frac{\pdv{F}{a}(a,b,p)}{\pdv{F}{p}(a,b,p)},
        \end{equation*}
        evaluates to zero when $p=0$, which follows from \eqref{eq:invariance-condition-2} provided $\pdv{F}{p}\neq0$ when $p=0$. Using a similar argument to that in the proof of Proposition \ref{proposition:smooth-safety-1}, we may exclude points satisfying $\pdv{F}{p}=0$ in the above, since such points cannot lie along any trajectory produced by the dynamics in \eqref{eq:p-dyn}. Thus, since both $\mathcal{H}_b^\varepsilon$ and $\R_{\geq0}$ are invariant for \eqref{eq:p-dyn} and the intersection of invariant sets is invariant \cite[Prop. 4.13]{Blanchini}, it follows that $\mathcal{H}_b^\varepsilon\cap\R_{\geq0}$ is invariant for \eqref{eq:p-dyn}, as desired. 
    \end{proof}
    Similar to Proposition \ref{proposition:smooth-safety-1}, the above result provides a simple condition, \eqref{eq:invariance-condition-2}, that one may use to help determine if an implicit function $\lambda$ satisfying \eqref{eq:F=0} and \eqref{eq:dFdp} will satisfy inequality \eqref{eq:lambda-ab}. Note that results similar to Lemma \ref{lemma:invariance} and Theorem \ref{theorem:smooth-safety} can be stated for Proposition \ref{proposition:smooth-safety-2}, the formal statements of which we omit here in the interest of space. As noted earlier, Sontag's $F$ in \eqref{eq:F-Sontag} satisfies the conditions of Proposition \ref{proposition:smooth-safety-2}, with $\varepsilon=2$. The following example introduces a Sontag-like safety filter that satisfies such conditions for any $\varepsilon\geq1$. 
    \begin{example}
        The smooth safety filters from Example \ref{ex:example-F} can approximate the QP-based safety filter \eqref{eq:k-CBF} arbitrarily closely. Generally, this is desirable; however, in certain situations, such as when handling uncertainty, one may wish to modulate how conservative the resulting safety filter is by tuning the value of $\varepsilon$ (cf. Fig. \ref{fig:main_figure}). For this, we introduce:
        \begin{equation}\label{eq:F-robust-Sontag}
            F(a,b,p) = bp^2 + \varepsilon a p - \tfrac{\varepsilon^2}{4}q(b),
        \end{equation}
        which generalizes both \eqref{eq:F-Sontag} and \eqref{eq:F-half-sontag} so that $F$ satisfies Proposition \ref{proposition:smooth-safety-2} for any $\varepsilon\geq1$ and produces a robust version of Sontag's formula $\lambda(a,b)=\tfrac{\varepsilon}{2}\lambda_{\mathrm{S}}(a,b)$ for any $\varepsilon\geq1$.
    \end{example}
    
    \begin{remark}
        Our approach allows for characterizing safety filters via $\varepsilon$. Although various choices of $F$ may contain tuning parameters,
        the behavior of such a safety filter is limited by the value of $\varepsilon$. When $\varepsilon=1$, $\lambda(a,b)$ may approach the boundary $a+b\lambda(a,b)=0$ of constraint \eqref{eq:lambda-ab}, whereas when $\varepsilon>1$, $\lambda(a,b)$ may only approach the boundary of the tightened constraint $a+\tfrac{1}{\varepsilon}b\lambda(a,b)\geq0$, resulting in a more conservative but also more robust safety filter (cf. Fig \ref{fig:simple-example}(b)).
    \end{remark}

    \section{Numerical Examples}\label{sec:examples}
    This section illustrates the practical benefits of smooth safety filters by applying 
    them to the model-free safety-critical control paradigm introduced in \cite{AmesRAL22}. Here, we design a safety filter for a reduced-order model, the safe trajectory of which is tracked by the full-order dynamics in a model-free fashion. We consider the same setting as in \cite[Ex. 2]{AmesRAL22}, which involves designing a controller for a planar Segway with configuration $(x,\varphi)\in\R^2$, where $x$ is the position and $\varphi$ the pitch angle, with the objective of driving forward at a desired velocity $\dot{x}_{\mathrm{d}}$ and stopping before colliding with a wall located at $x_{\max}$. This leads to the safety constraint $h(x) =x_{\max} - x$, which is used as a CBF for a one-dimensional single integrator to construct a safety filter $k_0\,:\,\R\rightarrow\R$ that produces a safe velocity for the Segway. On the full-order dynamics, this velocity is tracked by the PD controller:
    \begin{equation*}
        \bk(\bx)=K_{\mathrm{p}}(x -k_0(x)) + K_{\varphi}\varphi + K_{\dot{\varphi}}\dot{\varphi}
    \end{equation*}
    that also attempts to keep the Segway upright, where $\bx=(x,\varphi,\dot{x},\dot{\varphi})\in\R^4$ is the state and $K_{\mathrm{p}},K_{\varphi},K_{\dot{\varphi}}\in\R_{>0}$ are gains. Implementation of this controller does not require knowledge of the full-order dynamics, which may be uncertain or difficult to compute, and allows for enforcing \emph{input-to-state} safety \cite{AmesLCSS19} of the closed-loop system \cite[Prop. 1]{AmesRAL22}.

    We now compare the response of the Segway when the safe velocity is generated by a QP-based safety filter and the robust Sontag safety filter from \eqref{eq:F-robust-Sontag} for different values of $\varepsilon$. We begin by using the same parameters for the controller as in \cite[Ex. 2]{AmesRAL22}, the results of which are shown in Fig. \ref{fig:model-free-safety}(a). Here, both controllers safely track the reference velocity, and the response of the smooth controller approaches that of the QP controller as $\varepsilon\rightarrow1$. Although this approach does not directly rely on model knowledge, it relies on tuning the gains of the tracking controller to achieve safety. In general, safety can be achieved by increasing the proportional gain $K_{\mathrm{p}}$ to track the reference velocity more aggressively. When increasing $K_{\mathrm{p}}$ too much, however, we observe that the controller attempting to track a non-differentiable reference signal causes instabilities and safety violation\footnote{Note that the results in \cite{AmesRAL22} rely on differentiability of $k_0$.}. In contrast, the same controller attempting to track a smooth reference velocity is more oscillatory, but maintains safety. On the other extreme, taking $K_{\mathrm{p}}$ too low (see Fig. \ref{fig:model-free-safety}(c)) results in safety violation for the QP controller and smooth controller with $\varepsilon=1$, whereas controllers with $\varepsilon>1$ maintain safety due to their increased robustness.

    \begin{figure}
        \centering
        \vspace{0.25cm} 
        \includegraphics{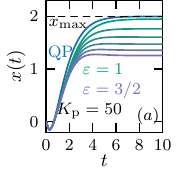}
        \hfill
        \includegraphics{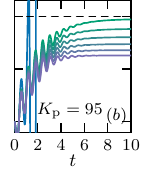}
        \hfill
        \includegraphics{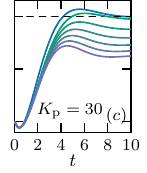}
        \hfill
        \vspace{-0.1cm}
        \includegraphics[width=0.48\textwidth]{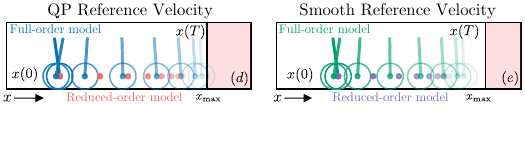}
        \vspace{-1.5cm}
        \caption{In each plot, the green/purple curves correspond to the trajectory obtained by tracking the smooth safe velocity generated by the robust Sontag safety filter for different $\varepsilon$ with $q(b)=0.1b$. Visualizations of trajectories from (c) are provided in (d) and (e), where the reduced-order trajectory corresponds to that of a single integrator $\dot{x}=u$ with $u=k_0(x)$ that the full-order model seeks to track.}
        \label{fig:model-free-safety}\vspace{-0.5cm}
    \end{figure}

    \section{Conclusion}
    We presented a general characterization of smooth safety filters based on the Implicit Function Theorem, leading to smooth universal formulas for safety-critical control that enable quantifying the conservatism of the resulting safety filter. The practical benefits of such smooth safety filters were showcased through their application to safety-critical control with reduced-order models \cite{AmesRAL22}. Future efforts will focus on extending our approach to multiple safety constraints and showcasing the benefits of smooth safety filters on hardware.
    
    \bibliographystyle{ieeetr}
    \bibliography{
        biblio/barrier,
        biblio/books,
        biblio/nonlinear,
        biblio/learning}

\end{document}